\documentclass[aps,prd,twocolumn,superscriptaddress,nofootinbib,floatfix,showpacs]{revtex4-2}

\usepackage[utf8]{inputenc}
\usepackage[T1]{fontenc}
\usepackage{amsmath,amssymb,amsfonts,mathtools}
\usepackage{amsthm}
\usepackage{bm}
\usepackage{graphicx}
\graphicspath{{figures/}}
\usepackage{xcolor}
\usepackage{hyperref}
\usepackage{orcidlink}

\hypersetup{colorlinks=true,
            linkcolor=blue!60!black,
            citecolor=blue!60!black,
            urlcolor=blue!60!black}
\providecommand{\orcidlink}[1]{%
  \href{https://orcid.org/#1}{\textsuperscript{\scriptsize ORCID}}}
\hypersetup{
  pdftitle={Bounded thermal weights from a discrete Boltzmann factor:
            spectral suppression, exact discrete Jarzynski identity,
            and the work-only obstruction},
  pdfauthor={Abdelmalek Boumali and Yassine Chargui}
}

\newtheorem{theorem}{Theorem}
\newtheorem{proposition}[theorem]{Proposition}
\newtheorem{lemma}[theorem]{Lemma}

\theoremstyle{remark}
\newtheorem*{remark}{Remark}

\newcommand{\be}{\begin{equation}}
\newcommand{\ee}{\end{equation}}
\newcommand{\bea}{\begin{align}}
\newcommand{\eea}{\end{align}}
\newcommand{\avg}[1]{\langle #1 \rangle}
\newcommand{\kB}{k_{\mathrm{B}}}
\newcommand{\Ep}{E_{\mathrm{P}}}
\newcommand{\lP}{\ell_{\mathrm{P}}}
\newcommand{\mP}{m_{\mathrm{P}}}
\newcommand{\Var}{\mathrm{Var}}
\newcommand{\dd}{\mathrm{d}}
\newcommand{\Wd}{\mathcal{W}_{d}}
\newcommand{\BE}{B_{E}}
\newcommand{\Zd}{Z_{d}}
\newcommand{\Fd}{F_{d}}
\newcommand{\Ud}{U_{d}}
\newcommand{\Ldisc}{L_{d}}
\newcommand{\LSB}{L_{\mathrm{SB}}}
\newcommand{\TBH}{T_{H}}
\newcommand{\bH}{\beta_{H}}
\newcommand{\nP}{n_{P}}

\begin{document}

\title{Bounded thermal weights from a discrete Boltzmann factor}

\author{Abdelmalek Boumali\,\orcidlink{0000-0003-2552-0427}}
\email{boumali.abdelmalek@univ-tebessa.dz}
\affiliation{Laboratory of Physics at Guelma (LPG),
             Echahid Cheikh Larbi Tebessi University,
             Tebessa 12000, Algeria}

\author{Yassine Chargui\,\orcidlink{0000-0001-9506-3159}}
\email{y.chergui@qu.edu.sa}
\affiliation{Department of Physics, College of Science, Qassim University,
             Buraydah 51452, Saudi Arabia}

\date{\today}

\begin{abstract}
We analyze the physical consequences of replacing the canonical
Boltzmann factor $e^{-\beta E}$ by the discrete weight
$\BE(\beta_{n})=(1-bE)^{n}$ defined on an inverse-temperature lattice
of spacing $b$. Positivity restricts the energy domain to $E<1/b$, so
the deformation endows the thermal weight with compact support;
this single feature organizes all subsequent results.
We obtain four main results.
\emph{(i)~Spectral.} An effective discrete Bose--Einstein occupation
law is derived, its regime of validity is bounded against the exact
truncated partition sum, and the relative error is shown to remain
below $2\%$ throughout the spectral region that dominates the
Stefan--Boltzmann integral.
\emph{(ii)~Hawking sector.} The leading correction to the
Schwarzschild luminosity reads
$\Ldisc/\LSB = 1 - (900\,\zeta(5)/\pi^{4})\,b\kB\TBH + \mathcal{O}(b^{2})$;
the Planck-scale calibration $b=8\pi/\Ep$ produces a Planckian endpoint
of the thermal description at $M_{\mathrm{cut}}=\mP$, which is not
interpreted as proof of an absolutely stable remnant.
\emph{(iii)~Nonequilibrium.} Defining a discrete work functional
$\Wd$ as a ratio of thermal weights, we prove an
\emph{exact} Jarzynski-type identity for deterministic
measure-preserving protocols. We further show that, at finite $b$, the
discrete Crooks ratio cannot be reduced to a function of work alone
because it retains explicit dependence on the initial trajectory
energy; the corresponding first-order Jensen inequality involves the
mixed moment $\avg{E_{i}W}$ rather than work cumulants.
\emph{(iv)~Entropy diagnostic.} A brick-wall computation produces a
logarithmic shift linear in $b$, whose normalization is explicitly
flagged as regulator-dependent.
The standard continuum results are recovered smoothly as $b\to 0$.
\end{abstract}

\maketitle

\section{Introduction}
\label{sec:intro}

The exponential Boltzmann factor $e^{-\beta E}$ is one of the central
mathematical structures of equilibrium statistical mechanics. It fixes
the canonical probability measure, determines the Planck distribution
of blackbody radiation, enters the thermal description of Hawking
emission, and provides the equilibrium kernel from which many
nonequilibrium fluctuation identities are constructed~\cite{huang1987,jarzynski1997,crooks1999,seifert2012,esposito2009}. For this reason,
any modification of the Boltzmann weight must be treated with
particular care. A change in the thermal factor is not a merely formal
replacement; it can modify spectral occupations, thermodynamic
response functions, high-energy emission rates, and the structure of
work relations. At the same time, the physical consequences of such a
modification should be separated from model-dependent assumptions, so
that the results that follow directly from the deformed weight are not
confused with interpretations that require an additional microscopic
theory.

Non-exponential thermal weights have been widely investigated in
systems where the assumptions leading to the Gibbs distribution are
modified or only effectively satisfied. Tsallis statistics and
superstatistics, for example, provide useful descriptions of systems
with long-range correlations, temperature fluctuations, memory effects,
or non-standard phase-space structures~\cite{tsallis1988,tsallis2009,beck2003,beck2001,wilk2000}. Such approaches have been applied to
gravitating systems~\cite{plastino1993}, high-energy particle spectra
\cite{bediaga2000,cms2010,alice2013,wong2013}, and dissipative
optical lattices~\cite{douglas2006}. Other departures from the
standard exponential law arise in contexts inspired by generalized
uncertainty principles~\cite{maggiore1993,kempf1995} and discrete-time
mechanics~\cite{lee1983,jaroszkiewicz1997}. These examples show that
the study of modified thermal weights is physically meaningful, but
they also show that each deformation has to be analyzed on its own
terms.

The present work examines a specific deformation proposed by Chung,
Hassanabadi, and Boumali~\cite{chung2022}, in which the inverse
temperature is defined on a discrete lattice,
$\mathcal{B}=\{\beta_n=\beta_0+nb\}$, and the ordinary Boltzmann factor
is replaced by
\be
\BE(\beta_n)=(1-bE)^n .
\label{eq:discBoltz}
\ee
Here $b$ is a positive parameter with dimensions of inverse energy, and
the continuum exponential is recovered in the correlated limit
$b\to0$, $n\to\infty$, with $nb=\beta$ fixed. The most important
feature of this deformation is not simply that it changes the
high-energy tail; rather, positivity imposes the condition $E<1/b$.
Thus the deformed thermal weight has compact support in energy. This
bounded support is the basic mechanism behind the spectral,
black-hole, and nonequilibrium effects studied in this paper.

Our objective is to determine which consequences follow directly from
this compactly supported thermal weight. The analysis is organized
around three connected questions. First, how does the discrete factor
modify the occupation of bosonic modes and the corresponding spectral
integrals? Second, what is the leading effect of this modification on
Hawking emission from a Schwarzschild black hole, and how should the
resulting endpoint of the thermal description be interpreted? Third,
how are nonequilibrium work identities modified when the thermal
weight is no longer exponential and therefore no longer multiplicative
in the energy difference? These questions are treated within the same
formal framework, which makes it possible to identify the precise role
played by the parameter $b$ in equilibrium, gravitational, and
nonequilibrium settings.

A central point of the paper is that the discrete weight must not be
treated as an exact finite-$b$ version of a geometric Boltzmann factor.
For a bosonic mode with occupation number $k$, one has
$(1-bk\varepsilon)^n\neq[(1-b\varepsilon)^n]^k$ at finite $b$.
Consequently, the usual derivation of the Bose--Einstein distribution
from a geometric series is not exact. We therefore introduce an
effective discrete Bose--Einstein occupation law and compare it with
the exact truncated partition sum. This comparison determines the
domain in which the effective law is reliable and shows that it
captures the part of the spectrum that dominates the Hawking
luminosity integral.

The black-hole application is then developed as a controlled
consequence of the modified spectral occupation. The bounded thermal
weight suppresses high-energy modes and yields a finite-$b$ correction
to the Stefan--Boltzmann luminosity. When the deformation scale is
calibrated at the Planck energy, the thermal channel terminates at a
Planckian mass. This endpoint should be understood as the limit of
validity of the thermal description generated by the discrete weight,
not as a proof of an absolutely stable remnant. The same distinction
is important in the entropy calculation: the brick-wall analysis gives
a logarithmic shift linear in $b$, but its numerical coefficient
depends on the chosen near-horizon regulator. We therefore use the
entropy result as a diagnostic of the deformation rather than as a
universal prediction.

The nonequilibrium part of the analysis reveals a different aspect of
the deformation. In the ordinary exponential case, the ratio of
thermal weights along a protocol depends on the work through
$e^{-\beta W}$. For the discrete weight, however, the natural ratio is
$(1-bE_f)^n/(1-bE_i)^n$, which depends separately on the initial and
final energies. This loss of exponential multiplicativity has two
important consequences. An exact Jarzynski-type identity can still be
proved for deterministic measure-preserving protocols, but the
corresponding Crooks ratio cannot be reduced to a function of work
alone at finite $b$. Moreover, the first-order Jensen inequality
contains the mixed moment $\langle E_i W\rangle$ rather than only work
cumulants. Hence the discrete framework preserves an integral
fluctuation identity while changing the statistical information needed
to characterize nonequilibrium transformations.

The main results of the paper can be summarized as follows. We derive
an effective discrete Bose--Einstein law and establish its validity by
comparison with the exact truncated partition function. We obtain the
leading finite-$b$ correction to the Schwarzschild luminosity and show
that, for the Planck-scale calibration $b=8\pi/\Ep$, the thermal
description reaches an endpoint at $M_{\mathrm{cut}}=\mP$. We prove an
exact discrete Jarzynski identity for the discrete work functional
$\Wd$ and show that no work-only Crooks ratio exists at finite $b$.
We also derive the corresponding first-order second-law inequality and
identify the mixed moment $\langle E_i W\rangle$ as the characteristic
correction. Finally, we compute the brick-wall entropy diagnostic and
make explicit its regulator dependence.

The paper is organized as follows. Section~\ref{sec:prelim}
introduces the discrete framework, the partition function, and the
effective Bose--Einstein law. Section~\ref{sec:BH} applies the formalism
to black-hole radiation, including the luminosity correction, the
Planckian endpoint, and the entropy diagnostic. Section~\ref{sec:fluctuation} develops the nonequilibrium theory, proves the
discrete Jarzynski identity, establishes the work-only obstruction, and
discusses possible observable implications. Section~\ref{sec:summary}
summarizes the conclusions and outlines open problems. The technical
details of the effective Bose--Einstein approximation and of the
brick-wall computation are given in Appendices~\ref{app:BE} and
\ref{app:entropy}, respectively.

\section{Discrete framework}
\label{sec:prelim}

We collect here the basic structures used throughout the paper. The
parameter $b$ has dimensions of inverse energy; $\hbar$ and $c$ are
kept explicit in dimensional formulas but set to unity in purely
statistical expressions; $\kB$ appears only when converting $\beta$
to thermodynamic temperature, $T=1/(\kB\beta)$. We distinguish
carefully between exact identities (valid at any finite $b$) and
controlled effective approximations valid for $b\to 0$ in the
correlated continuum limit.

\subsection{Discrete Boltzmann factor and continuum limit}

The inverse-temperature lattice is
\be
\beta_{n}=\beta_{0}+nb,\qquad n=0,1,2,\dots,
\label{eq:lattice}
\ee
with $b>0$; we set $\beta_{0}=0$ throughout. The discrete Boltzmann
factor is~\cite{chung2022}
\be
\BE(\beta_{n})=(1-bE)^{n},
\label{eq:BF}
\ee
which is positive iff $0\le E<1/b$, defining the energy cutoff
\be
E_{\max}=1/b.
\label{eq:Emax}
\ee
The standard exponential is recovered through the correlated limit
\be
\lim_{\substack{b\to 0,\;n\to\infty\\ nb=\beta}}(1-bE)^{n}
   = e^{-\beta E},
\label{eq:contlim}
\ee
which is the elementary identity $(1-x/n)^{n}\to e^{-x}$ at $x=\beta E$.
At finite $b$, the weight obeys the lattice recursion
\be
\BE(\beta_{n+1})-\BE(\beta_{n})=-bE\,\BE(\beta_{n}),
\label{eq:lattice_recursion}
\ee
the forward-difference analogue of
$\partial_{\beta}e^{-\beta E}=-E e^{-\beta E}$. The deformation thus
replaces continuous cooling in $\beta$ by a stepwise evolution along
the lattice, with the compact-support condition acting already at the
level of this finite-difference dynamics.

\paragraph{Lattice vs continuum statements.}
Exact statements are made at integer $n$. When we write $n=\beta/b$
and expand in $b$, we are taking the correlated large-$n$ continuum
approximation with $\beta$ fixed. Several results below are
\emph{exact} for the lattice formulation, whereas others are
controlled effective approximations valid at small $b$.

\paragraph{Comparison with Tsallis statistics.}
The Tsallis $q$-exponential~\cite{tsallis1988} is
\be
e_{q}(x)=[1+(1-q)x]^{1/(1-q)},\qquad q\neq 1.
\label{eq:qexp}
\ee
Setting $q=1-bE$ formally relates $e_{q}(-\beta E)$ to
Eq.~\eqref{eq:BF}, but the two are \emph{not} identical: the Tsallis
exponent is the continuous real number $1/(1-q)$, whereas the discrete
exponent is the non-negative integer $n$. The relation is therefore
an asymptotic correspondence in the correlated continuum limit; the
discrete framework should not be presented as an exact finite-$b$
realization of Tsallis statistics.

\subsection{Discrete partition function and free energy}

For a system with discrete energies $\{E_{k}\}$ and degeneracies
$g_{k}$, the discrete partition function is
\be
\Zd(n)=\sum_{k:\,E_{k}<1/b} g_{k}\,(1-bE_{k})^{n},
\label{eq:Zd}
\ee
and the discrete free energy~\cite{chung2022} reads
\be
\Fd=-\frac{1}{b}\,\ln\!\left[\frac{\Zd(n+1)}{\Zd(n)}\right].
\label{eq:Fd}
\ee
For two protocol values $\lambda_{i}$, $\lambda_{f}$, the
corresponding discrete free-energy difference $\Delta\Fd$ is fixed by
the equivalent natural identity
\be
\bigl(1-b\,\Delta\Fd\bigr)^{n}\equiv
\frac{Z_{d}^{(f)}}{Z_{d}^{(i)}},
\label{eq:DeltaFd_def}
\ee
which agrees with $-(1/\beta)\ln(Z^{(f)}/Z^{(i)})$ in the continuum
limit but differs from it at $\mathcal{O}(b)$. We use
Eq.~\eqref{eq:DeltaFd_def} as the operational definition of
$\Delta\Fd$ throughout, since it is the form that appears in the
exact discrete Jarzynski identity (Theorem~\ref{thm:jarz}).
The discrete internal energy is
\be
\Ud=\frac{1}{b}\!\left[1-\frac{\Zd(n+1)}{\Zd(n)}\right]
   =\frac{1}{b}\bigl[1-e^{-b\Fd}\bigr],
\label{eq:Ud}
\ee
and the discrete heat capacity is the second finite difference,
\be
C_{d}^{(b)}
   =\frac{\Ud(n+1)-\Ud(n)}{T_{n+1}-T_{n}}
   =\frac{b\,\Delta_{n}\Ud}
        {1/(\kB\beta_{n+1})-1/(\kB\beta_{n})}.
\label{eq:Cd}
\ee
In the continuum limit, $\Ud\to\avg{E}$ and
$C_{d}^{(b)}\to\kB\beta^{2}\Var(E)$, recovering the standard
fluctuation--dissipation relation. At finite $b$, however, compact
support implies $\Ud\le 1/b$ and the sign of $C_{d}^{(b)}$ remains
model-dependent.

The bounded support changes the qualitative structure of the canonical
ensemble: every thermodynamic sum is taken over a compact domain, so
all moments of $E$ are finite by construction. The high-temperature
regime is then governed by how the occupation weight reorganizes
within $0\le E<1/b$, which is the common origin of the spectral,
Hawking, and nonequilibrium modifications discussed below.

\subsection{Effective discrete Bose--Einstein law}
\label{subsec:BE}

The single-mode partition function for a bosonic mode of energy
$\varepsilon$ is the truncated sum
\be
\Zd^{\mathrm{BE,exact}}(\varepsilon)
   =\sum_{k=0}^{k_{\max}}(1-b\,k\varepsilon)^{n},
   \quad k_{\max}=\lfloor 1/(b\varepsilon)\rfloor.
\label{eq:ZBE_exact}
\ee
Unlike $e^{-\beta k\varepsilon}=(e^{-\beta\varepsilon})^{k}$, the
discrete weight does \emph{not} factorize:
$(1-bk\varepsilon)^{n}\neq[(1-b\varepsilon)^{n}]^{k}$ for $k\ge 2$.
Eq.~\eqref{eq:ZBE_exact} therefore admits no closed-form geometric
sum at finite $b$. A direct expansion gives
\be
(1-bk\varepsilon)^{n}
  =[(1-b\varepsilon)^{n}]^{k}
   \Bigl[1+\tfrac{1}{2}nb^{2}\varepsilon^{2}k(k-1)+\mathcal{O}(b^{3})\Bigr],
\label{eq:factor_err}
\ee
so factorization is restored to leading order whenever
$nb^{2}\varepsilon^{2}=b\beta\varepsilon\cdot b\varepsilon\ll 1$.

Motivated by this, we \emph{define} the effective discrete
Bose--Einstein occupation by formal analogy with the geometric-series
closure that holds in the exponential case,
\be
\avg{n_{\varepsilon}}_{d}
  \equiv\frac{1}{(1-b\varepsilon)^{-n}-1}
\label{eq:nBE}
\ee
This definition has three properties that justify its use below:
(i)~it reduces to the Planck distribution in the continuum limit;
(ii)~it incorporates the single-mode cutoff at $\varepsilon=1/b$;
(iii)~it captures the leading $\mathcal{O}(b)$ correction to the
exact occupation throughout the Wien regime.

\begin{lemma}[Validity of the effective BE law]
\label{lem:BEvalid}
For $b\varepsilon\le 0.01$ and $\beta\varepsilon\ge 1$, the
relative error
$\bigl|\avg{n_{\varepsilon}}_{d}-\avg{k}_{\mathrm{exact}}\bigr|/
 \avg{k}_{\mathrm{exact}}$
is bounded by $2\%$. In the Rayleigh--Jeans regime
$\beta\varepsilon\lesssim 0.3$, the effective law underestimates the
true suppression by a factor $\sim 4\nP+1$, where
$\nP=1/(e^{\beta\varepsilon}-1)$.
\end{lemma}
\noindent The proof, by direct comparison of the $\mathcal{O}(b)$
expansions of the exact and effective occupations, is given in
Appendix~\ref{app:BE}. The relevance of Lemma~\ref{lem:BEvalid} is
that the Hawking luminosity integral is dominated by modes with
$\beta\varepsilon\gtrsim 1$ (cf.\ Sec.~\ref{subsec:Hawking}), so the
deep-infrared limitation does not affect (R2).

Writing $n=\beta/b$ and expanding to first order in $b$,
\be
\avg{n_{\varepsilon}}_{d}
  \approx\frac{1}{e^{\beta\varepsilon}-1}
   -\frac{b\beta\varepsilon^{2}}{2}\,
    \frac{e^{\beta\varepsilon}}{(e^{\beta\varepsilon}-1)^{2}}
   +\mathcal{O}(b^{2}).
\label{eq:nBEexpand}
\ee
The correction is negative for all $\varepsilon>0$, confirming that
the effective discrete distribution \emph{systematically suppresses}
high-energy occupation. Two points are worth emphasizing.
Eq.~\eqref{eq:nBE} should be read as an effective single-mode law,
justified by Lemma~\ref{lem:BEvalid}, rather than as an exact
finite-$b$ identity. The $\mathcal{O}(b)$ correction in
Eq.~\eqref{eq:nBEexpand} scales as $b\beta\varepsilon^{2}$ well before
the cutoff is reached, so discreteness effects switch on smoothly
rather than only at the endpoint.

Figure~\ref{fig:planck} illustrates this behavior on a logarithmic
vertical scale, using the dimensionless combination
$x=\beta\varepsilon$ for a representative ratio $b/\beta$. Three
features are simultaneously visible: (i)~near-continuum behavior at
small $x$; (ii)~progressive suppression at intermediate $x$; and
(iii)~exact truncation at $x_{\max}=\beta/b$.

\begin{figure*}[t]
\centering
\includegraphics[width=0.85\textwidth]{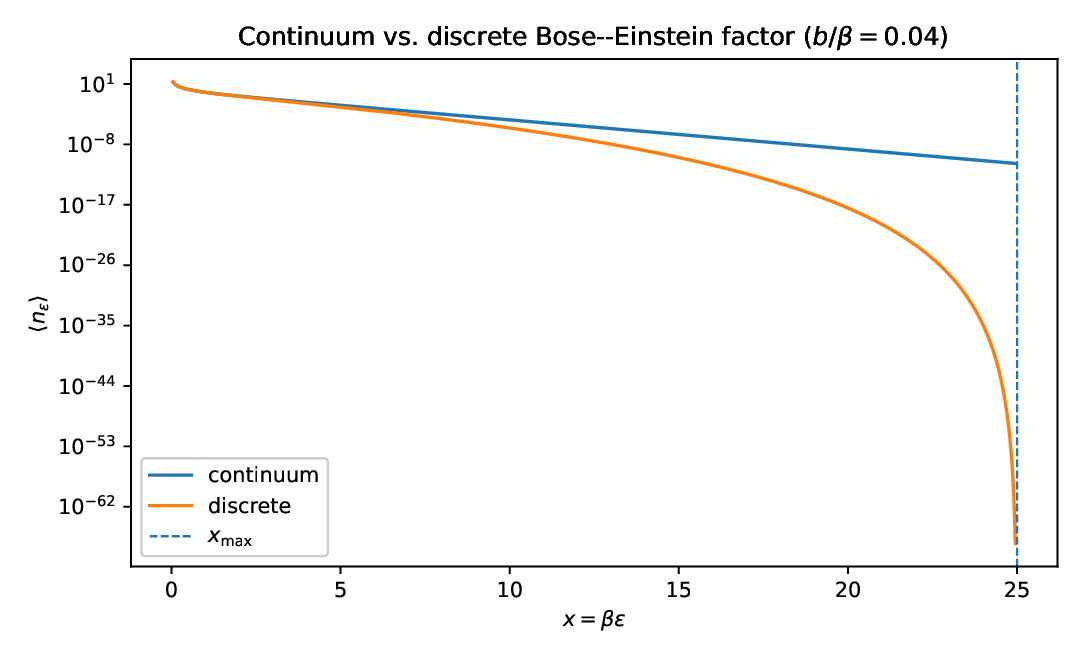}
\caption{Dimensionless comparison of the continuum Bose--Einstein
factor and its discrete counterpart, plotted against
$x=\beta\varepsilon$ for $b/\beta=0.04$ ($n=25$). The vertical axis
is logarithmic so the infrared overlap and the ultraviolet
suppression are visible in a single panel. Three features control
the results below: near-continuum behavior in the infrared, stronger
suppression at intermediate $x$, and exact truncation of the
spectrum at $x_{\max}=\beta/b$.}
\label{fig:planck}
\end{figure*}

\section{Black-hole sector}
\label{sec:BH}

We now apply the discrete framework to the most direct gravitational
arena: the Hawking spectrum, in which thermal occupation numbers
appear explicitly. Compact support of the discrete weight feeds
directly into the spectrum, the integrated luminosity, and the
interpretation of the terminal stage of evaporation.

For a Schwarzschild black hole of mass $M$, the Hawking
temperature~\cite{hawking1975} is
\be
\TBH=\frac{\hbar c^{3}}{8\pi GM\kB},\qquad
\bH=\frac{8\pi GM}{\hbar c^{3}}.
\label{eq:TH}
\ee
Placing $\bH$ on the lattice $\beta_{n}=nb$ and identifying the mass
spacing $\delta M$ with the lattice quantum---as motivated by
Bekenstein's area
quantization~\cite{bekenstein1973,bekenstein1974} and by loop quantum
gravity~\cite{rovelli2004,thiemann2007,dreyer2003}---fixes
\be
b=\frac{8\pi G\,\delta M}{\hbar c^{3}}.
\label{eq:bBH}
\ee
The Planck-scale calibration $\delta M=\mP=(\hbar c/G)^{1/2}$ then
gives
\be
b=\frac{8\pi}{\Ep},
\label{eq:bPlanck}
\ee
with $\Ep=\mP c^{2}\approx 1.22\times 10^{19}\,\mathrm{GeV}$. The
energy cutoff is then
$E_{\max}=\Ep/(8\pi)\sim 5\times 10^{17}\,\mathrm{GeV}$. The integer
$n=M/\delta M$ ranges from $n\sim 10^{38}$ for a solar-mass black hole
down to $n=1$ at the cutoff.

\begin{remark}
The identification of $b$ with a Planckian mass spacing is a
phenomenological calibration of the deformation scale, not an
independent derivation from a microscopic theory of quantum gravity.
The statistical construction itself only requires $b>0$ and $E<1/b$.
\end{remark}

\subsection{Spectrum, luminosity, and thermal cutoff}
\label{subsec:Hawking}

Substituting Eq.~\eqref{eq:nBE} into the Planckian Hawking spectrum
yields the discrete occupation
\be
\avg{n_{\omega}}_{d}=\frac{1}{(1-b\hbar\omega)^{-n}-1}.
\label{eq:Hspec}
\ee
We discuss its key features in turn.

\paragraph{Ultraviolet cutoff.}
The spectrum vanishes for $\omega\ge\omega_{\max}=1/(b\hbar)$. At
Planck-scale $b$, this removes trans-Planckian modes from the
effective thermal spectrum. The cutoff should be read as a
regularization of the thermal channel only, not as a complete
solution of trans-Planckian issues in dynamical
backgrounds~\cite{jacobson1991,unruh1995,martin2001}.

\paragraph{Infrared agreement.}
For $b\hbar\omega\ll 1$, expansion of Eq.~\eqref{eq:Hspec} reproduces
the Planck spectrum at leading order, with a negative subleading
correction
\be
\avg{n_{\omega}}_{d}
   \approx\frac{1}{e^{\bH\hbar\omega}-1}
   -\frac{b\bH(\hbar\omega)^{2}}{2}\,
    \frac{e^{\bH\hbar\omega}}{(e^{\bH\hbar\omega}-1)^{2}}.
\label{eq:IRagreement}
\ee

\paragraph{Grey-body modification.}
With a frequency-dependent grey-body factor $\Gamma_{\ell}(\omega)$
for partial wave $\ell$, the differential luminosity per mode reads
\be
\frac{\dd L_{\ell}}{\dd\omega}
   =\frac{\Gamma_{\ell}(\omega)\,\hbar\omega^{3}}{2\pi}\,
    \frac{1}{(1-b\hbar\omega)^{-n}-1}.
\label{eq:greybody}
\ee
The deformation modifies only the occupation factor; standard
grey-body corrections still encode propagation through the curvature
potential. In the analytic estimate below we use the constant-
absorption $s$-wave approximation, which keeps the coefficient
transparent. Frequency-dependent factors~\cite{page1976} would alter
the numerical prefactor but neither the compact-support statement nor
the sign of the leading depletion.

\paragraph{Spin-dependent emission.}
For fermions,
$\avg{n_{\omega}}_{d,\mathrm{FD}}=1/[(1-b\hbar\omega)^{-n}+1]$, with
the same UV cutoff. The fermionic-to-bosonic ratio is
\be
\frac{\avg{n_{\omega}}_{d,\mathrm{FD}}}{\avg{n_{\omega}}_{d,\mathrm{BE}}}
   =\tanh\!\Bigl[\tfrac{n}{2}\ln\bigl(1/(1-b\hbar\omega)\bigr)\Bigr],
\label{eq:FD_ratio}
\ee
which reduces to $\tanh(\beta\hbar\omega/2)$ as $b\to 0$.

\subsubsection*{Luminosity correction}

In the constant-absorption $s$-wave approximation, the total emitted
power is
\be
\Ldisc=\frac{\sigma_{\mathrm{abs}}}{2\pi^{2}}
   \int_{0}^{1/(b\hbar)}\!\!
   \frac{\hbar\omega^{3}\,\dd\omega}{(1-b\hbar\omega)^{-n}-1}.
\label{eq:Ld}
\ee
Using the expansion
\be
\frac{1}{(1-b\hbar\omega)^{-n}-1}
  =\frac{1}{e^{\bH\hbar\omega}-1}
   -\frac{b\bH(\hbar\omega)^{2}\,e^{\bH\hbar\omega}}
        {2(e^{\bH\hbar\omega}-1)^{2}}
   +\mathcal{O}(b^{2}),
\label{eq:integrand_expand}
\ee
substituting into Eq.~\eqref{eq:Ld}, and applying the integral identity
\be
\int_{0}^{\infty}\!\!\dd x\,\frac{x^{5}\,e^{x}}{(e^{x}-1)^{2}}
  =120\,\zeta(5),
\label{eq:zeta5}
\ee
we find
\be
\Ldisc=\LSB\!\left[1-\frac{900\,\zeta(5)}{\pi^{4}}\,b\kB\TBH
   +\mathcal{O}(b^{2})\right]\!,
\label{eq:Ld_expand}
\ee
where $\LSB=\sigma_{\mathrm{abs}}\pi^{2}\kB^{4}\TBH^{4}/(120\hbar^{3})$
and $\zeta(5)\approx 1.037$. Numerically the coefficient is
$900\zeta(5)/\pi^{4}\approx 9.58$. The negative sign confirms that the
discrete spectrum radiates \emph{less power} than the Planckian one at
the same Hawking temperature.

Physically, Eq.~\eqref{eq:Ld_expand} expresses the fact that the
compact-support deformation systematically depletes the ultraviolet
modes that dominate the Stefan--Boltzmann integral. The numerical
coefficient is not universal (it depends on the constant-absorption
$s$-wave approximation), but the negative sign and the linear scaling
with $b\kB\TBH$ are robust. This is one of the clearest quantitative
consequences that follows directly from the thermal deformation,
independent of any additional assumption about microscopic gravity.

\subsection{Planckian endpoint of the thermal description}
\label{subsec:evap}

The mass-loss rate $\dd M/\dd t=-\Ldisc/c^{2}$ combined with
Eq.~\eqref{eq:Ld_expand} gives
\be
\frac{\dd M}{\dd t}
   =-\frac{C}{M^{2}}\!
    \left[1-\frac{900\,\zeta(5)}{\pi^{4}}\,
                  \frac{b\hbar c^{3}}{8\pi GM}+\mathcal{O}(b^{2})\right]\!,
\label{eq:evap}
\ee
where $C>0$. The correction reduces the thermal evaporation rate as
$M$ approaches the cutoff scale, but Eq.~\eqref{eq:evap} is a
first-order result in $b\kB\TBH$ and should not be extrapolated all
the way to the Planck regime.

A more conservative nonperturbative statement follows from the
lattice itself. With $n=M/\delta M$ and $\delta M=\mP$, the smallest
positive lattice value $n=1$ corresponds to
\be
M_{\mathrm{cut}}=\frac{\hbar c^{3}b}{8\pi G}=\delta M=\mP,
\label{eq:remnant}
\ee
equivalently $\kB\TBH=1/b$. The discrete \emph{thermal description}
therefore reaches a Planckian endpoint when the characteristic
Hawking energy is of the same order as the compact-support scale.

\begin{remark}
Eq.~\eqref{eq:remnant} does not establish the existence of an
absolutely stable remnant. It states only that the continuum thermal
evaporation picture has reached the boundary of the finite-$b$
description: the spectrum is strongly depleted, the energy domain is
compact, and the mass lattice has no smaller positive point.
Nonthermal tunneling, pair creation beyond the occupation-number
picture, and other quantum-gravitational
processes~\cite{almheiri2013,penington2020} could still be relevant.
\end{remark}

The mechanism here is therefore comparable to remnant
scenarios~\cite{adler2001,scardigli1999,cavaglia2004,chen2015} in
which evaporation slows near a Planckian scale. Compared with
generalized-uncertainty-principle and noncommutative
approaches~\cite{nicolini2006,smailagic2003}, the distinguishing
feature here is that the slowdown is traced directly to the compact
support of the thermal weight, rather than to a modified uncertainty
relation or a deformed near-horizon geometry. Establishing an
actually stable remnant would require a separate dynamical analysis
beyond the thermal sector, including the degeneracy issues
emphasized in the remnant literature~\cite{chen2015,unruh2017}.

\subsection{Brick-wall entropy diagnostic (non-universal)}
\label{subsec:entropy}

The entropy discussion must be interpreted with more caution than the
luminosity calculation. The aim here is not a microscopic derivation
of black-hole entropy from the discrete weight, but the narrower
question: if the usual near-horizon thermal mode count is deformed by
the bounded occupation factor, what type of subleading correction is
induced in a standard brick-wall computation?

Replacing the continuum thermal weight in the brick-wall free energy
by the discrete one (see Appendix~\ref{app:entropy} for details), the
leading divergence still renormalizes the area term, and the first
discreteness correction generates an additional logarithmic piece.
A representative first-order result is
\be
S^{(d)}_{\mathrm{BH}}
  =\frac{A}{4\lP^{2}}
   +\bigl(\alpha_{0}+\delta\alpha_{b}^{(\mathrm{bw})}\bigr)
    \,\ln\!\left(\frac{A}{\lP^{2}}\right)+s_{0}+\mathcal{O}(b^{2}),
\label{eq:SBH}
\ee
with $\alpha_{0}=-1/2$ recovering the standard continuum
value~\cite{carlip2000,kaul2000,meissner2004} and
$\delta\alpha_{b}^{(\mathrm{bw})}\propto b$ a regulator-dependent
discreteness shift. The structural message is that a bounded thermal
weight induces a logarithmic correction \emph{linear in $b$}; the
explicit normalization, derived in Appendix~\ref{app:entropy} under
specific brick-wall conventions, should be read as an
order-of-magnitude indicator rather than a universal prediction. We
do \emph{not} promote the numerical value of
$\delta\alpha_{b}^{(\mathrm{bw})}$ to the status of a physical result.

In the continuum limit only the standard logarithmic term with
$\alpha_{0}=-1/2$ survives. At finite $b$, any specific value of
$\delta\alpha_{b}^{(\mathrm{bw})}$ obtained within the brick-wall
scheme should be compared with coefficients produced by loop quantum
gravity~\cite{kaul2000,meissner2004} and conformal
techniques~\cite{carlip2000} as a consistency check rather than as a
prediction. What the estimate provides is a diagnostic: the linear
dependence on $b$ originates from the same low-order expansion that
controls the spectral suppression, while the numerical coefficient
inherits the regulator-sensitive content of the near-horizon density
of states.

\section{Nonequilibrium work relations}
\label{sec:fluctuation}

We now turn to the nonequilibrium sector, which probes the
statistical structure most directly. The Jarzynski
equality~\cite{jarzynski1997} and the Crooks
theorem~\cite{crooks1999} link nonequilibrium work measurements to
equilibrium free-energy differences and underpin the second law at
the fluctuation level~\cite{seifert2012,esposito2009}. The standard
derivations rely on the multiplicativity
$e^{-\beta W}e^{-\beta E_{i}}=e^{-\beta(W+E_{i})}$, which $(1-bx)^{n}$
does not satisfy. We construct the discrete generalization
accordingly.

\subsection{Discrete work functional and exact Jarzynski identity}
\label{subsec:Jarz}

For a trajectory $\gamma:i\to f$ connecting an initial microstate of
energy $E_{i}^{(\lambda_{i})}$ to a final microstate of energy
$E_{f}^{(\lambda_{f})}$, define the \emph{discrete work functional}
\be
\Wd[\gamma]
  \equiv
  \left(\frac{1-bE_{f}^{(\lambda_{f})}}
             {1-bE_{i}^{(\lambda_{i})}}\right)^{\!n}.
\label{eq:Wd}
\ee
In the continuum limit $b\to 0$ with $nb=\beta$ fixed,
$\Wd\to e^{-\beta W}$ where $W=E_{f}^{(\lambda_{f})}-E_{i}^{(\lambda_{i})}$.

\begin{theorem}[Exact discrete Jarzynski identity]
\label{thm:jarz}
Let the driving protocol be implemented by a deterministic bijection
(hence measure-preserving)
$\Phi:i\mapsto f=\Phi(i)$ on the finite state space $\{E_{k}<1/b\}$,
with initial distribution
\be
p_{i}^{(d)}=\frac{(1-bE_{i}^{(\lambda_{i})})^{n}}{Z_{d}^{(i)}}.
\label{eq:initial_dist}
\ee
Then
\be
\bigl\langle\Wd\bigr\rangle
  =\frac{Z_{d}^{(f)}}{Z_{d}^{(i)}}
  \equiv (1-b\,\Delta\Fd)^{n},
\label{eq:discJarz}
\ee
where $\Delta\Fd$ is the discrete free-energy difference defined
through Eq.~\eqref{eq:Fd}.
\end{theorem}

\begin{proof}
By definition,
\be
\avg{\Wd}
  =\sum_{i}p_{i}^{(d)}\,\Wd[i\!\to\!\Phi(i)]
  =\frac{1}{Z_{d}^{(i)}}\sum_{i}(1-bE_{\Phi(i)}^{(\lambda_{f})})^{n}.
\label{eq:proof_step1}
\ee
Since $\Phi$ is a bijection of the finite state space, the substitution
$f=\Phi(i)$ rewrites the sum without reweighting,
\be
\avg{\Wd}
  =\frac{1}{Z_{d}^{(i)}}\sum_{f}(1-bE_{f}^{(\lambda_{f})})^{n}
  =\frac{Z_{d}^{(f)}}{Z_{d}^{(i)}}.
\label{eq:proof_step2}
\ee
The final equality follows from Eq.~\eqref{eq:Fd}.
\end{proof}

Equation~\eqref{eq:discJarz} is the exact identity that survives the
loss of multiplicative exponential structure. Instead of averaging
$e^{-\beta W}$, one averages a ratio of discrete weights between
trajectory endpoints: the natural object measures how the statistical
weight assigned to a trajectory changes under the protocol when
inverse temperature lives on a lattice.

The deterministic measure-preserving assumption is the discrete
analogue of the setting in which the ordinary Jarzynski equality
admits its simplest proof. We sketch a stochastic generalization
below.

\paragraph{First-order expansion.}
Expanding $\ln\Wd$ to first order in $b$,
\be
\ln\Wd
  =-\beta W-\frac{b\beta}{2}\bigl(2E_{i}W+W^{2}\bigr)
   +\mathcal{O}(b^{2}).
\label{eq:Wd_expand}
\ee
The extra term depends separately on the trajectory-specific initial
energy $E_{i}$, not only on the work $W$. Consequently, it is
\emph{not} legitimate to replace the exact ratio $\Wd$ by the simpler
$(1-bW)^{n}$: the latter differs already at order $b$. In particular,
$\avg{(1-bW)^{n}}=(1-b\,\Delta F)^{n}+\mathcal{O}(b^{2})$ is
\emph{not} valid in general; the discrepancy is of order $b$, not
$b^{2}$. Only the full ratio $\Wd$ satisfies the exact identity. This
is structural: because $(1-bE)^{n}$ lacks the multiplicative property
of $e^{-\beta E}$, work alone is no longer sufficient to describe the
deformation---one must keep the trajectory endpoints explicitly.

\subsection{Stochastic generalization (sketch)}
\label{subsec:stoch}

Theorem~\ref{thm:jarz} extends naturally to stochastic dynamics
satisfying a discrete form of local detailed balance. Let
$T_{i\to f}^{(\lambda)}$ be the transition kernel under instantaneous
protocol value $\lambda$, with stationary measure
$p_{i}^{(d,\lambda)}=(1-bE_{i}^{(\lambda)})^{n}/Z_{d}^{(\lambda)}$.
Detailed balance with respect to the discrete weight reads
\be
\frac{T_{i\to f}^{(\lambda)}}{T_{f\to i}^{(\lambda)}}
   =\left(\frac{1-bE_{f}^{(\lambda)}}{1-bE_{i}^{(\lambda)}}\right)^{\!n}.
\label{eq:DBd}
\ee
For a discrete-time protocol $\lambda_{0}\to\lambda_{1}\to\cdots\to\lambda_{N}$
generating a stochastic trajectory $\gamma=(i_{0},i_{1},\dots,i_{N})$,
the natural multi-step work functional is the telescoping product
\be
\Wd[\gamma]
   =\prod_{k=0}^{N-1}
    \left(\frac{1-bE_{i_{k+1}}^{(\lambda_{k+1})}}
               {1-bE_{i_{k+1}}^{(\lambda_{k})}}\right)^{\!n},
\label{eq:Wd_stoch}
\ee
which for a deterministic bijection collapses to Eq.~\eqref{eq:Wd}.
A standard trajectory-summation argument, using
Eq.~\eqref{eq:DBd} at each protocol step, yields
$\avg{\Wd}=Z_{d}^{(\lambda_{N})}/Z_{d}^{(\lambda_{0})}=(1-b\,\Delta\Fd)^{n}$,
recovering Theorem~\ref{thm:jarz}. We do not develop the full
construction here; the point is that the algebraic structure required
is precisely Eq.~\eqref{eq:DBd}, with no further use of exponential
factorization.

\subsection{Discrete Crooks ratio: the work-only obstruction}
\label{subsec:Crooks}

For forward ($F$) and time-reversed ($R$) protocols, the path-weight
ratio is
\be
\frac{P_{F}[\gamma]}{P_{R}[\bar\gamma]}
   =\frac{(1-bE_{i})^{n}}{(1-bE_{f})^{n}}
    \cdot\frac{Z_{d}^{(f)}}{Z_{d}^{(i)}}.
\label{eq:Crooks_path}
\ee
In the exponential case,
$e^{-\beta E_{i}}/e^{-\beta E_{f}}=e^{\beta W}$ depends only on $W$.
The discrete ratio, by contrast, depends on $E_{i}$ separately.

\begin{proposition}[No work-only Crooks at finite $b$]
\label{prop:nowork}
For any $b>0$, the path-weight ratio $P_{F}[\gamma]/P_{R}[\bar\gamma]$
in Eq.~\eqref{eq:Crooks_path} is \emph{not} a function of $W[\gamma]$
alone. Equivalently, there exist pairs of trajectories $\gamma_{1}$,
$\gamma_{2}$ with $W[\gamma_{1}]=W[\gamma_{2}]$ but
$P_{F}[\gamma_{1}]/P_{R}[\bar\gamma_{1}]
 \neq P_{F}[\gamma_{2}]/P_{R}[\bar\gamma_{2}]$.
\end{proposition}

\begin{proof}
Expanding the logarithm of Eq.~\eqref{eq:Crooks_path} to first order
in $b$ gives
\be
\ln\!\frac{P_{F}[\gamma]}{P_{R}[\bar\gamma]}
   =\beta(W-\Delta F)
   +b\beta E_{i}W+\tfrac{b\beta}{2}(W^{2}-\Delta F^{2})
   +\mathcal{O}(b^{2}),
\label{eq:discCrooks}
\ee
in which $E_{i}$ refers to the initial energy of the specific
trajectory~$\gamma$. At fixed $W\neq 0$, varying $E_{i}$ changes the
right-hand side at order $b$, which proves the claim.
\end{proof}

If one nevertheless wishes to compare with the standard Crooks
formulation, the observable distributions are
\be
P_{F}(W)=\sum_{\gamma:W[\gamma]=W}P_{F}[\gamma],
\label{eq:PFW}
\ee
and similarly for $P_{R}(-W)$. The ratio $P_{F}(W)/P_{R}(-W)$ involves
a conditional average of $(1+b\beta E_{i}W)$ over trajectories
producing work $W$, and reduces to a clean analytic expression only
under additional assumptions about the joint distribution of
$(E_{i},W)$.

\paragraph{Worked example: two-level system.}
Consider a single two-level system with energies
$\{0,\varepsilon(\lambda)\}$, whose gap is driven from
$\varepsilon_{i}=\varepsilon(\lambda_{i})$ to
$\varepsilon_{f}=\varepsilon(\lambda_{f})$ by a deterministic
bijection that maps each level of the initial Hamiltonian to the
corresponding level of the final one. The two trajectories have work
values $W_{1}=0$ (ground branch) and $W_{2}=\varepsilon_{f}-\varepsilon_{i}$
(excited branch). The discrete partition functions are
$Z_{d}^{(i)}=1+(1-b\varepsilon_{i})^{n}$ and
$Z_{d}^{(f)}=1+(1-b\varepsilon_{f})^{n}$, and Theorem~\ref{thm:jarz}
is verified by direct evaluation. For the excited branch,
\be
\frac{P_{F}[\gamma_{2}]}{P_{R}[\bar\gamma_{2}]}
   =\left(\frac{1-b\varepsilon_{i}}{1-b\varepsilon_{f}}\right)^{\!n}
    \frac{Z_{d}^{(f)}}{Z_{d}^{(i)}},
\label{eq:toy_ratio}
\ee
whose continuum limit is $e^{\beta(W_{2}-\Delta F)}$ as expected.
Expanding the logarithm to first order in $b$ and using $\beta=nb$,
\be
\begin{split}
\ln\!\frac{P_{F}[\gamma_{2}]}{P_{R}[\bar\gamma_{2}]}
  &=\beta(W_{2}-\Delta\Fd)\\
  &\quad+\frac{b\beta}{2}(2\varepsilon_{i}W_{2}+W_{2}^{2}-\Delta\Fd^{2})
   +\mathcal{O}(b^{2}),
\end{split}
\label{eq:toy_log}
\ee
matching Eq.~\eqref{eq:discCrooks} with $E_{i}=\varepsilon_{i}$. The
correction $b\beta\varepsilon_{i}W_{2}$ is the clean manifestation of
the work-only obstruction: keeping $W_{2}$ fixed and varying
$\varepsilon_{i}$ changes the right-hand side at order $b$. Direct
numerical evaluation of Eqs.~\eqref{eq:toy_ratio} and~\eqref{eq:toy_log}
for $b=10^{-2}$, $n=20$, $\varepsilon_{i}\in[0.5,10]$, $W_{2}=1$
confirms a linear deviation with slope $b\beta\varepsilon_{i}W_{2}$
and residuals of order $b^{2}$.

\begin{figure*}[t]
\centering
\includegraphics[width=0.95\textwidth]{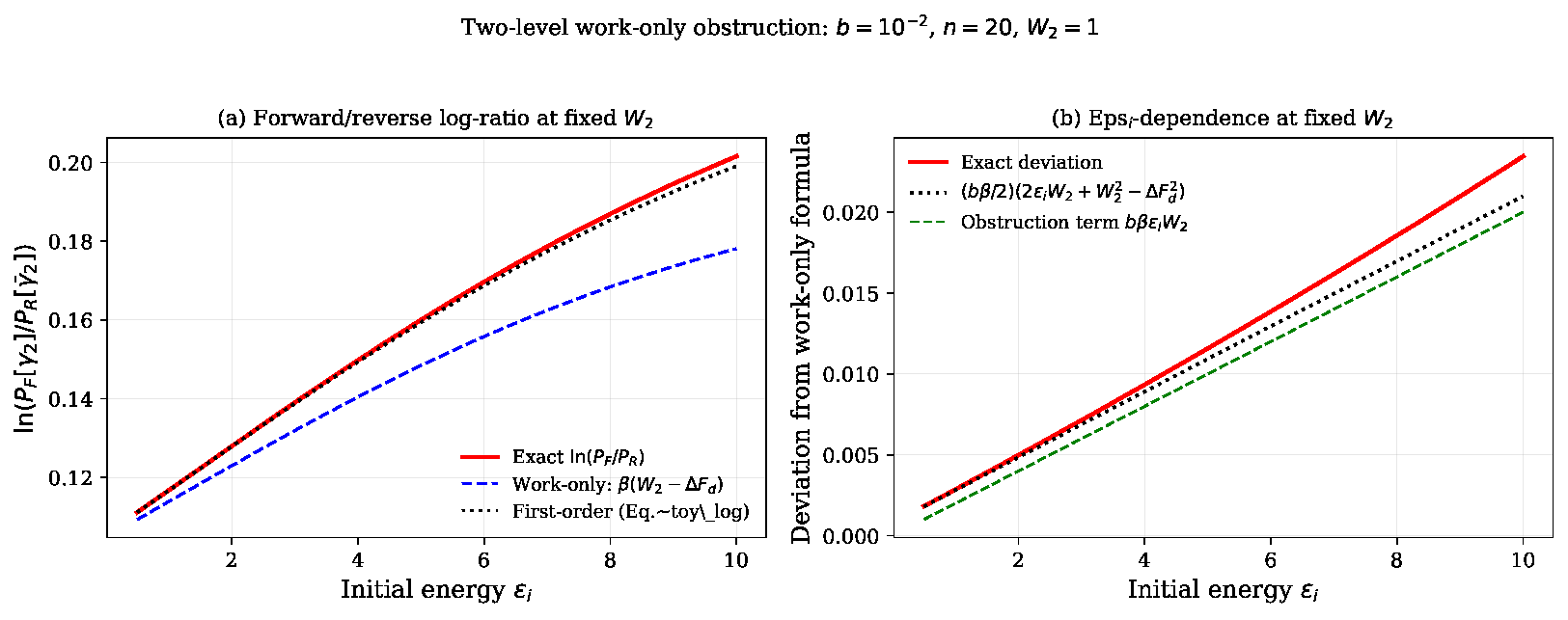}
\caption{Numerical demonstration of the work-only obstruction in the
two-level system, for $b=10^{-2}$, $n=20$, $W_{2}=1$.
\textbf{(a)}~The exact log path-weight ratio in
Eq.~\eqref{eq:toy_ratio} (red), the work-only continuum prediction
$\beta(W_{2}-\Delta\Fd)$ (blue dashed), and the first-order discrete
expression of Eq.~\eqref{eq:toy_log} (black dotted), all plotted
against the initial energy $\varepsilon_{i}$ at fixed $W_{2}$. The
exact and first-order curves agree to high accuracy throughout, while
the work-only term diverges from them at large $\varepsilon_{i}$.
\textbf{(b)}~Deviation from the work-only formula, isolating the
obstruction. The exact deviation (red) grows linearly with
$\varepsilon_{i}$ at slope $b\beta W_{2}=2\times 10^{-3}$, in
quantitative agreement with the leading obstruction term
$b\beta\varepsilon_{i}W_{2}$ (green dashed) and with the full
first-order prediction (black dotted). The figure makes
Proposition~\ref{prop:nowork} concrete: at finite $b$, two
trajectories with the same work but different initial energy do not
have the same forward/reverse path-weight ratio.}
\label{fig:crooks}
\end{figure*}

\subsection{First-order Jensen inequality}
\label{subsec:2ndlaw}

A valid finite-$b$ inequality follows from applying Jensen's
inequality to the exact identity Eq.~\eqref{eq:discJarz}. Since
$-\ln x$ is convex for $x>0$,
\be
\avg{-\ln\Wd}\ge -\ln\avg{\Wd}=-n\ln(1-b\,\Delta F).
\label{eq:jensen_exact}
\ee
Using the expansion Eq.~\eqref{eq:Wd_expand}, the left-hand side reads
\be
\avg{-\ln\Wd}
  =\beta\avg{W}+b\beta\avg{E_{i}W}
   +\tfrac{b\beta}{2}\avg{W^{2}}+\mathcal{O}(b^{2}),
\label{eq:LHS_jensen}
\ee
while the right-hand side reads
\be
-n\ln(1-b\,\Delta F)
   =\beta\Delta F+\tfrac{b\beta}{2}\Delta F^{2}+\mathcal{O}(b^{2}).
\label{eq:RHS_jensen}
\ee
Combining and rearranging yields the first-order-in-$b$ inequality
\be
\avg{W}\ge\Delta F-b\avg{E_{i}W}
   -\tfrac{b}{2}\bigl(\avg{W^{2}}-\Delta F^{2}\bigr)
   +\mathcal{O}(b^{2}),
\label{eq:2ndlaw}
\ee
valid under the same assumptions as Theorem~\ref{thm:jarz}. Three
points are noteworthy. The zeroth-order content is the standard
$\avg{W}\ge\Delta F$. The $\mathcal{O}(b)$ correction depends on the
mixed moment $\avg{E_{i}W}$ and the second moment $\avg{W^{2}}$, so
it does not reduce to a function of the variance or skewness of $W$
alone. Finally, unlike a cumulant-truncation ``second-law correction''
of the form $\avg{W}\ge\Delta F+(\beta/2)\Var(W)$, which is valid only
in the Gaussian limit, Eq.~\eqref{eq:2ndlaw} follows directly from
Jensen and from the exact identity Eq.~\eqref{eq:discJarz} without
cumulant truncation.

The mixed moment $\avg{E_{i}W}$ is not an artifact of the derivation
but a structural feature of the discrete framework, consistent with
Proposition~\ref{prop:nowork}.

\subsection{Observable consequences and experimental regimes}
\label{subsec:experiment}

The first-order correction in Eq.~\eqref{eq:discCrooks} suggests that
the natural scale of the finite-$b$ deviation from the standard
Crooks relation is
\be
\bigl|\delta\ln(P_{F}/P_{R})\bigr|\sim b\beta\avg{E_{i}W}.
\label{eq:scale}
\ee
The interpretation depends on the meaning of $b$.

\paragraph{Universal Planck-suppressed parameter.}
For $b\sim 10^{-28}\,\mathrm{eV}^{-1}$ (Planck-suppressed) and
$\avg{W}\sim\avg{E}\sim\kB T$ at room temperature in single-molecule
experiments~\cite{liphardt2002}, the order-of-magnitude estimate of
the fractional deviation is $b\beta\avg{E_{i}W}\sim 10^{-27}$, far
below experimental resolution. Fluctuation-theorem experiments are
then conceptually important but not competitive as direct probes of
fundamental quantum gravity.

\paragraph{Analogue or effective realization.}
Bounded weights may also arise in engineered or coarse-grained systems
with an emergent $b_{\mathrm{eff}}$ unrelated to the universal Planck
scale~\cite{collin2005}. Much larger apparent shifts could then be
realized, and the discrete Jarzynski framework would become testable
as an effective statistical theory. Because the exact relation
retains explicit $E_{i}$ dependence, a quantitative comparison with
data would require either a protocol controlling the $E_{i}$--$W$
correlation, or a direct measurement of $\Wd$ rather than only the
work distribution.

\paragraph{Estimator bias.}
The first-order expansion of Eq.~\eqref{eq:discJarz} implies that the
modified Jarzynski estimator carries a systematic bias of order
\be
\bigl|\Delta F^{(d)}_{\mathrm{est}}-\Delta F^{(0)}_{\mathrm{est}}\bigr|
   \sim b\,\avg{W^{2}}+\mathcal{O}(b^{2}).
\label{eq:bias}
\ee
For Planck-suppressed $b$ this bias is negligible; in analogue
implementations with larger $b_{\mathrm{eff}}$ it may become relevant,
and its exact coefficient depends on the joint distribution of
$(E_{i},W)$.

\section{Summary and outlook}
\label{sec:summary}

We have analyzed a single structural feature of the discrete
Boltzmann factor $\BE(\beta_{n})=(1-bE)^{n}$: unlike the canonical
weight, it has compact support in energy. Once that point is kept at
the center of the discussion, the consequences fall into a clear
hierarchy.

\emph{Thermal sector.} The bounded weight suppresses occupation
numbers progressively before the cutoff and removes them entirely at
$E=1/b$. In the Hawking application, this depletes the spectrum in
its intermediate and ultraviolet parts, lowers the integrated
luminosity by the universal-sign correction
$-(900\zeta(5)/\pi^{4})b\kB\TBH$, and brings the thermal description to
a Planckian endpoint at $M_{\mathrm{cut}}=\mP$ for the Planck-scale
calibration $b=8\pi/\Ep$. Whether the associated terminal
configuration corresponds to an absolutely stable remnant, a
long-lived quasi-static phase, or a regime in which nonthermal
channels dominate is a question the thermal framework alone cannot
settle.

\emph{Nonequilibrium sector.} Defining the discrete work functional
as a ratio of thermal weights yields an exact Jarzynski-type identity
for deterministic measure-preserving protocols
(Theorem~\ref{thm:jarz}), which extends naturally to stochastic
dynamics satisfying a discrete local detailed balance
(Sec.~\ref{subsec:stoch}). The Crooks ratio cannot be a function of
work alone at finite $b$ (Proposition~\ref{prop:nowork}), and the
correct first-order Jensen inequality involves the mixed moment
$\avg{E_{i}W}$ rather than work cumulants
[Eq.~\eqref{eq:2ndlaw}]. The loss of exponential multiplicativity
forces trajectory endpoints to be tracked explicitly---a structural
consequence of the deformation that is coherent with the thermal
analysis.

\emph{Entropy sector.} A brick-wall computation produces a
logarithmic shift linear in $b$, but its normalization inherits the
regulator dependence of the scheme; it should be read as a
diagnostic, not as a universal prediction.

\emph{Experimental status.} For universal Planck-suppressed $b$,
direct corrections to fluctuation theorems are unobservable. Analogue
or engineered systems with emergent $b_{\mathrm{eff}}$ could test the
formal structure of the discrete nonequilibrium theory without
probing fundamental quantum gravity.

Several questions remain open.
(i)~A microscopic derivation of the entropy correction bypassing the
brick-wall regulator.
(ii)~A complete construction of the stochastic discrete Jarzynski
identity sketched in Sec.~\ref{subsec:stoch}, including dissipation
identities and a discrete fluctuation theorem for entropy
production.
(iii)~Extension of the two-level worked example to multi-level systems
and continuous driving, where the conditional distribution
$p(E_{i}\mid W)$ becomes nontrivial and a closed-form
distribution-level Crooks generalization may be possible.
(iv)~A systematic comparison between $b$ and the deformation scales
of generalized uncertainty principle
models~\cite{maggiore1993,kempf1995,scardigli1999} and other
quantum-gravity-inspired frameworks.
(v)~Translation of bounds on the extragalactic CMB
spectrum~\cite{fixsen1996} and the projected sensitivity of
next-generation spectrometers such as PIXIE~\cite{kogut2011} into
upper bounds on $b$, an avenue that is in principle open but
quantitatively challenging given the Planck suppression expected in
the fundamental interpretation.

\begin{acknowledgments}
The author thanks colleagues at the Laboratory of Physics at Guelma
for useful discussions and acknowledges the institutional support of
Echahid Cheikh Larbi Tebessi University.
\end{acknowledgments}

\appendix

\section{Comparison of the effective Bose--Einstein law with the exact
truncated partition sum}
\label{app:BE}

This appendix proves Lemma~\ref{lem:BEvalid} by comparing the
$\mathcal{O}(b)$ expansions of the exact and effective occupation
numbers, and explains why the deep-infrared limitation of the
effective law does not affect the Hawking luminosity integral.

\subsection{Exact $\mathcal{O}(b)$ correction}

Expanding the discrete single-state weight $(1-bk\varepsilon)^{n}$
around the exponential gives
\be
(1-bk\varepsilon)^{n}=e^{-\beta k\varepsilon}
   \Bigl[1-\tfrac{1}{2}b\beta\varepsilon^{2}k^{2}+\mathcal{O}(b^{2})\Bigr],
\label{eq:weight_exp}
\ee
with $\beta=nb$. The associated correction to the exact occupation
number is
\be
\Delta\avg{k}_{\mathrm{exact}}
  \equiv\avg{k}_{\mathrm{exact}}-\nP
  =-\tfrac{1}{2}b\beta\varepsilon^{2}
   \bigl[\avg{k^{3}}_{P}-\nP\avg{k^{2}}_{P}\bigr]
   +\mathcal{O}(b^{2}),
\label{eq:exact_Ob}
\ee
where averages are taken with the standard Planck distribution and
$\nP=1/(e^{\beta\varepsilon}-1)$. Using
$\avg{k^{2}}_{P}=\nP(2\nP+1)$ and
$\avg{k^{3}}_{P}=\nP(6\nP^{2}+6\nP+1)$, this reduces to
\be
\Delta\avg{k}_{\mathrm{exact}}
  =-\tfrac{1}{2}b\beta\varepsilon^{2}\,\nP(\nP+1)(4\nP+1)
   +\mathcal{O}(b^{2}).
\label{eq:exact_explicit}
\ee

\subsection{Effective $\mathcal{O}(b)$ correction}

Expanding Eq.~\eqref{eq:nBE} similarly gives
\be
\Delta\avg{k}_{\mathrm{eff}}
  =-\tfrac{1}{2}b\beta\varepsilon^{2}\,\nP(\nP+1)+\mathcal{O}(b^{2}),
\label{eq:eff_explicit}
\ee
since
$e^{\beta\varepsilon}/(e^{\beta\varepsilon}-1)^{2}=\nP(\nP+1)$. The
two expressions therefore differ by a factor $(4\nP+1)$:
\begin{itemize}
\item In the Wien regime $\beta\varepsilon\gg 1$, $\nP\ll 1$ and
$(4\nP+1)\to 1$: the effective law reproduces the exact correction.
\item In the Rayleigh--Jeans regime $\beta\varepsilon\ll 1$,
$(4\nP+1)\sim 4/(\beta\varepsilon)$ grows: the effective law
underestimates the suppression.
\end{itemize}
This proves the second clause of Lemma~\ref{lem:BEvalid}; the first
clause follows from explicit numerical comparison of
Eqs.~\eqref{eq:nBE} and~\eqref{eq:ZBE_exact}, summarized in
Fig.~\ref{fig:bevsexact}(a).

\subsection{Weight in the luminosity integral}

Figure~\ref{fig:bevsexact}(b) shows the standard Planck integrand
$x^{3}/(e^{x}-1)$ on the same horizontal axis. Quantitative
integration shows that $96.5\%$ of the integral is accumulated in
$x\ge 1$ and $99.9\%$ in $x\ge 0.3$, so the infrared region where the
effective law fails contributes negligibly to the total Hawking flux.
This justifies using Eq.~\eqref{eq:nBE} inside the luminosity integral
of Sec.~\ref{subsec:Hawking}: the dominant contribution lives in
exactly the regime where the effective law is accurate.

\begin{figure*}[t]
\centering
\includegraphics[width=0.85\textwidth]{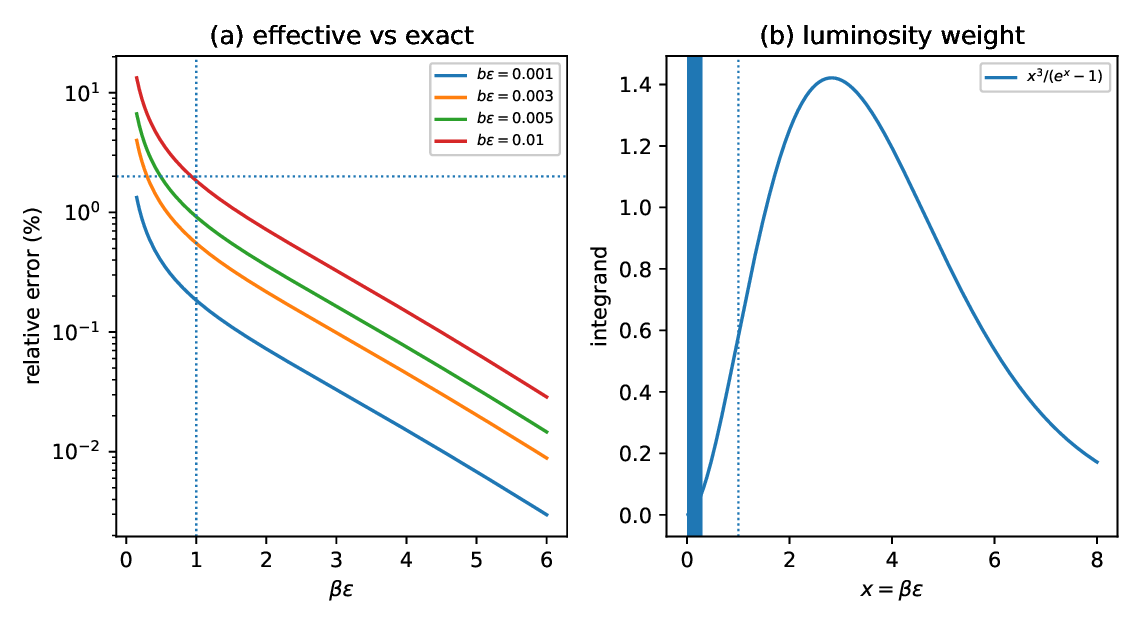}
\caption{(a)~Relative error of the effective discrete Bose--Einstein
law~\eqref{eq:nBE} with respect to the exact truncated partition
sum~\eqref{eq:ZBE_exact}, plotted against $\beta\varepsilon$ for
representative values of $b\varepsilon$. The vertical red dotted line
marks the $2\%$ level; the gray dotted line marks
$\beta\varepsilon=1$. (b)~Planck luminosity integrand $x^{3}/(e^{x}-1)$
on the same horizontal axis, with the shaded region
$\beta\varepsilon\le 0.3$ indicating where the effective law is
inaccurate. The integrand is strongly suppressed in this region, which
is why the luminosity integral in Sec.~\ref{subsec:Hawking} is
controlled by the effective law despite its limitations in the deep
infrared.}
\label{fig:bevsexact}
\end{figure*}

\section{Sketch of the brick-wall entropy estimate}
\label{app:entropy}

This appendix records the assumptions and intermediate steps behind
Eq.~\eqref{eq:SBH}. The aim is not to claim a microscopic derivation
but to make explicit where the linear-in-$b$ logarithmic shift
originates.

\subsection{Discrete brick-wall free energy}

Starting from the bosonic free-energy expression in a black-hole
background, we replace the standard thermal weight by the discrete
one, obtaining schematically
\be
\Fd\approx\frac{1}{\bH}\int_{0}^{1/(b\hbar)}\!\dd\omega\,g(\omega)\,
   \ln\!\bigl[1-(1-b\hbar\omega)^{n}\bigr],
\label{eq:Fd_bw}
\ee
where $g(\omega)$ is the brick-wall density of states and the upper
limit implements the compact support of the discrete factor. For
$\bH=nb$ and $b\hbar\omega\ll 1$,
\be
\begin{split}
\ln\!\bigl[1-(1-b\hbar\omega)^{n}\bigr]
   &=\ln(1-e^{-\bH\hbar\omega})\\
   &\quad+\frac{b\bH}{2}\,
          \frac{(\hbar\omega)^{2}}{e^{\bH\hbar\omega}-1}
   +\mathcal{O}(b^{2}),
\end{split}
\label{eq:logexp}
\ee
so $\Fd=F_{0}+\delta F_{b}+\mathcal{O}(b^{2})$ with
\be
\delta F_{b}\propto b\int\!\dd\omega\,g(\omega)\,
  \frac{(\hbar\omega)^{2}}{e^{\bH\hbar\omega}-1}.
\label{eq:dFb}
\ee

\subsection{Density of states and entropy operation}

The near-horizon density of states contains an area-law term and
subleading pieces:
\be
g(\omega)=c_{3}\,A\,\omega^{2}+c_{1}\,\omega^{0}+\cdots,
\label{eq:gomega}
\ee
where $c_{3}$ and $c_{1}$ depend on the radial cutoff and on the
background geometry. The $\omega^{2}$ piece renormalizes the area
term, while the $\omega^{0}$ piece yields a logarithmic contribution
under the standard entropy operation
\be
S=\bH^{2}\frac{\partial F}{\partial\bH}.
\label{eq:SfromF}
\ee
Using the standard integrals
\be
\int_{0}^{\infty}\!\!\dd x\,\frac{x^{3}}{e^{x}-1}=\frac{\pi^{4}}{15},
\quad
\int_{0}^{\infty}\!\!\dd x\,\frac{x}{e^{x}-1}=\frac{\pi^{2}}{6},
\label{eq:integrals_app}
\ee
the discreteness correction multiplies the same geometric structure
that produces the standard logarithm. With the conventional
brick-wall normalization adopted in Sec.~\ref{subsec:entropy},
\be
\delta\alpha_{b}^{(\mathrm{bw})}=\frac{b\Ep}{8\pi}.
\label{eq:alphabw_app}
\ee
Because $c_{3}$ and $c_{1}$ are regulator-sensitive and
Eq.~\eqref{eq:logexp} is only a first-order deformation of a
regulated calculation, this should be read as a representative
estimate of a linear-in-$b$ shift, not as a universal prediction.


\end{document}